\documentclass[a4paper,12pt]{article}
\usepackage[utf8]{inputenc}
\usepackage[english]{babel}

\usepackage{amsmath,amssymb,amsfonts,amsthm}
\usepackage{latexsym,bm,textcomp}

\usepackage{graphicx}
\usepackage{xcolor}
\usepackage{hyperref}
\usepackage{grffile}
\usepackage{bmpsize}
\AtBeginDocument{\DeclareGraphicsExtensions{.pdf,.PDF,.eps,.EPS,.png,.PNG,.tif,.TIF,.jpg,.JPG,.jpeg,.JPEG}}

\DeclareGraphicsExtensions{.pdf,.png,.jpg}
\usepackage{longtable,tabulary}
\usepackage{booktabs,array,multirow}

\usepackage{fullpage}
\usepackage{setspace}
\usepackage{parskip}
\usepackage{float}
\usepackage[section]{placeins}
\usepackage{titlesec}

\usepackage[numbers,sort&compress]{natbib}

\usepackage{authblk}

\theoremstyle{plain}
\theoremstyle{definition} 
\newtheorem{thm}{Theorem}[section]
\newtheorem{lemma}[thm]{Lemma}

\theoremstyle{definition}

\newtheorem{remark}{Remark}

\renewenvironment{abstract}
  {{\bfseries\noindent{\abstractname}\par\nobreak}\footnotesize}
  {\bigskip}

\titlespacing{\section}{0pt}{*3}{*1}
\titlespacing{\subsection}{0pt}{*2}{*0.5}
\titlespacing{\subsubsection}{0pt}{*1.5}{0pt}

\newif\iflatexml\latexmlfalse
\providecommand\citet{\cite}
\providecommand\citep{\cite}

\begin{document}
\title{
Spectrum for a non-unitary one-dimensional two-state quantum walk with one defect}

\author[1]{Takako Endo}%
\author[2]{Yohei Matsumoto}%
\author[3]{Hiromichi Ohno}%
\author[4]{Akito Suzuki}%
\affil[1]{Institute for Materials Research (IMR), Tohoku University, Sendai, Miyagi 980-8577, Japan}
\affil[2]{Department of Science and Technology, Graduate School of Medicine, Science and Technology, Shinshu University, Wakasato,Nagano 380-8553, Japan}
\affil[3]{Department of Mathematics, Faculty of Engineering, Shinshu University, Wakasato,Nagano 380-8553, Japan}%
\affil[4]{Department of Production Systems Engineering and Sciences, Komatsu University, Komatsu, Ishikawa 923-8511, Japan}%

\vspace{-1em}

  \date{\today}

\begingroup
\let\center\flushleft
\let\endcenter\endflushleft
\maketitle
\endgroup

\selectlanguage{english}
\begin{abstract}
Existence of the eigenvalues of the discrete-time quantum walks is deeply related to localization.
Also, for the study of open quantum systems, non-Hermitian systems have attracted much attention. As mathematical models for such systems, non-unitary quantum walks with the chiral symmetry are essential for the study of the topological insulator. In this paper, we give the whole picture of the eigenvalues of a non-unitary one-dimensional two-state quantum walks with one defect and the chiral symmetry.
\end{abstract}

\sloppy

\section{Introduction}
In recent years, the quantum walk , which is regarded as the quantum mechanical analogue of the classical random walk, has attracted considerable attention\cite{AharonovEtAl2001, AmbainisEtAl2001, Kempe2003, Kendon2007, QiangEtAl2024,VAndraca2012}.  

Quantum walks, which exploit quantum superposition and interference effects, exhibit propagation speeds surpassing those of classical diffusion and have been extensively studied from both algorithmic and physical perspectives \cite{Meyer1997}.
They have attained universality as a general computational model, demonstrating that any quantum circuit can be efficiently simulated on a suitably constructed graph \cite{ChildsEtAl2013,ChildsEtAl2009, LovettEtAl2010}.
Building upon this universality, research has advanced toward designing new quantum algorithms that leverage the structural properties of quantum walks \cite{Ambainis2004, Szegedy2004}.
At the same time, quantum walks have been applied to the modeling of physical systems, enabling the simulation of topological phase transitions and interface states, thereby contributing to quantum simulations of topological materials \cite{KitagawaEtAl2012,PanahiyanEtAl2019}. Furthermore, quantum walks form a unified framework for quantum simulation itself, establishing a direct connection between discrete-time dynamics and continuous-space quantum evolution \cite{Yepez2005}.
Moreover, non-unitary quantum walks have been employed for simulating natural phenomena, providing a quantum-mechanical explanation for the high efficiency of excitonic energy transfer in photosynthetic complexes \cite{MohseniEtAl2008}.
Thus, quantum walks continue to evolve as a central concept in quantum information science, encompassing computational theory, algorithmic design, and the simulation of natural processes.

One important application of (discrete-time) quantum walk is Grover's search algorithm~\cite{Grover1996}. Shenvi \textit{et al.}~\cite{ShenviEtAl2003} demonstrated that one of the models of quantum walk is equivalent to Grover’s search algorithm. Moreover, quantum walks have the following two properties, they are considered useful for generalizing Grover's search algorithm. One is that the spread speed of quantum walk is fast. Indeed, while the width of the probability distribution in a classical random walk diverges as $\mathcal{O}(\sqrt{t})$, the width of quantum walks diverges linearly with time $t$~\cite{KnightEtAl2003}. The other is localization. Localization is defined as the phenomenon in which the probability of finding the walker at a certain position does not converge to zero even in the long-time limit.  
It is known that localization occurs when the time-evolution operator of the quantum walk possesses eigenvalues~\cite{InuiEtAl2004}.

Many studies have investigated both the enhancement of search speed and the emergence of localization.  
In the following, we briefly review previous works related to localization and the eigenvalue problem. Konno~\cite{Konno2010} studied a quantum walk model with a perturbation at the origin, referred to as a one defect unitary quantum walk, and proved the occurrence of localization using a path-counting method. W\'ojcik \textit{et al.}~\cite{Antoni2012} introduced a one defect quantum walk model:
\begin{align}
C(x)=
\begin{cases}
\dfrac{1}{\sqrt{2}}
\begin{bmatrix}
1 & 1\\
1 & -1
\end{bmatrix} & (x = \pm1, \pm2, \ldots), \\[10pt]
\dfrac{\omega}{\sqrt{2}}
\begin{bmatrix}
1 & 1\\
1 & -1
\end{bmatrix} & (x = 0),
\end{cases}
\nonumber
\end{align}
where $\omega = e^{2i\pi\phi}$ and $\phi \in (0,1)$.  
This model describes a globally unitary system with a local perturbation $\omega$ at the origin. Their analytical results showed that localization can arise even from such a minimal phase defect.  
They also demonstrated that the existence of localization depends strongly on the initial coin state.  
Following this, Endo \textit{et al.}~\cite{EndoKonno2014} used the generating function method, and derived the candidate eigenvalues of this system. Furthermore, Endo \textit{et al.}~\cite{KonnoEtAl2013} generalized this result to a two-phase quantum walk with one defect and obtained candidate eigenvalues.  
Later, Kiumi \textit{et al.}~\cite{KiumiSaito2021} extended this model to a more general class of two-phase quantum walks with one defect and derived a necessary and sufficient condition for the existence of eigenvalues using the transfer matrix method.  
They also obtained explicit expressions for eigenvalues and eigenvectors in special cases.  
This transfer matrix approach provides a powerful framework for including eigenvalues that could not be accessed by the generating function method.

Recent studies have also investigated non-unitary quantum walks, whose time-evolution operators are no longer unitary.  
For instance, Mochizuki \textit{et al.}~\cite{Mochizuki2016} introduced a model of non-unitary quantum walk and analyzed the $\mathcal{PT}$-symmetry of this non-unitary quantum walk. Furthermore, Asahara et al.~\cite{AsaharaEtAl2021} derived the Witten index and essential spectrum for this model. Additionally, Kiumi et al.~\cite{KiumiEtAl2022} derived the eigenvalues of the two-phase quantum walk, a special case of this model, using the transfer matrix method.

Motivated by these developments, the present work aims to solve completely the eigenvalue problem of a non-unitary quantum walk model that possesses a one defect and chiral symmetry, using the transfer matrix method.  
This model can be regarded as a non-unitary extension of the unitary model studied in~\cite{EndoKonno2014,Antoni2012}.  
In the unitary case, the spectrum lies on the unit circle; however, due to the non-unitarity in our model, the spectrum extends beyond the unit circle, exhibiting a spiderweb-like structure.

The structure of this paper is as follows:  
Section~2 introduces the definition of our quantum walk model.  
Section~3 presents the main results.  
Section~4 defines the transfer matrix and provides the lemmas and proofs needed to establish the main theorem.

\section{Definition of our non-unitary model}

A quantum walk considered in this paper is a one-dimensional two-state quantum walk with one defect.
Let
\begin{eqnarray}
\mathcal{H} =\ell ^2(\mathbb{Z} ;\mathbb{C} ^2)=\Biggl\{ \Psi : \mathbb{Z} \to \mathbb{C} ^2  ~\Bigg{|}~ 
\sum_{x \in \mathbb{Z}} \Arrowvert \Psi (x)\Arrowvert ^2 _{\mathbb{C}^2} <\infty \ \Biggr\} \nonumber
\end{eqnarray}
be the Hilbert space of states. $\mathcal{H}$ is naturally identified with $\ell^2 (\mathbb{Z})\oplus \ell^2 (\mathbb{Z})$ and $\bigoplus _{x \in \mathbb{Z}}\mathbb{C}^2$. For a vector $\Psi \in \mathcal{H}$ and $x \in \mathbb{Z}$, we denote
\begin{eqnarray}
\Psi(x)= \left[	
  \begin{array}{cc}
  \Psi _{L}(x)  \\
  \Psi _{R}(x)  \\ 
  \end{array}
  \right]\ \in \mathbb{C}^2. \nonumber
\end{eqnarray}
The time-evolution operator $U$ of a one-dimensional two-state quantum walk is defined as the product $U=SC$ of two operators $S$ and $C$ on $\mathcal{H}$. Here, $S$ is a shift operator defined by
\begin{eqnarray}
S=\left[	
  \begin{array}{cc}
  L &0 \\
   0 & L^*\\ 
  \end{array}
  \right] \nonumber
  \end{eqnarray}
on $\mathcal{H}$ identified with $\ell^2(\mathbb{Z}) \oplus \ell^2(\mathbb{Z})$, where $L$ is the shift operator on $\ell^2(\mathbb{Z})$ defined by $(Lf)(x)=f(x+1)$ $(x\in \mathbb{Z})$ for $f\in \ell^2(\mathbb{Z})$. The operator $C$ is a coin operator defined by 
\begin{eqnarray}
C= \bigoplus_{x \in \mathbb{Z}}C(x)\nonumber
  \end{eqnarray}
on $\mathcal{H}$ identified with $\bigoplus _{x \in \mathbb{Z}}\mathbb{C}^2$, where $\{C(x)\}_{x\in \mathbb{Z}}$ is a family of two by two matrices. In this paper, we take $C(x)$ as
\begin{eqnarray}
C(x)=\left\{ \begin{array}{ll}
\dfrac{1}{\sqrt{2}}\begin{bmatrix} 1&-1\\ 1& 1\end{bmatrix}& (x=\pm1,\pm2,\cdots), \\\\
\dfrac{\omega}{\sqrt{2}}\begin{bmatrix} 1&-1\\ 1& 1\end{bmatrix}& (x=0) \\
\end{array} \right. \nonumber
\end{eqnarray}
with a parameter $\omega \in \mathbb{R}\backslash\{0\}$, and we denote the time-evolution operator $U$ as 
\begin{eqnarray}
U=U_\omega  \quad (\omega \in \mathbb{R}\backslash\{0\}) \nonumber
\end{eqnarray}
to emphasize the dependence of the parameter $\omega$. When $\omega =1$, we write $C=C_1$. Since
\begin{eqnarray}
C_1(x)=\begin{array}{ll}
\dfrac{1}{\sqrt{2}}\begin{bmatrix}
1&-1\\ 
1& 1
\end{bmatrix} 
\end{array} \nonumber
\end{eqnarray}
for all $x \in {\mathbb Z}$, 
the time-evolution operator $U_1$ is homogeneous. When $\omega \neq 1$, the coin operator 
\begin{eqnarray}
C(x)=\begin{array}{ll}
\dfrac{1}{\sqrt{2}}\begin{bmatrix}
1&-1\\ 
1& 1
\end{bmatrix}
\end{array} 
+(\omega -1)\delta_{0}(x)
\begin{array}{ll}
\dfrac{1}{\sqrt{2}}\begin{bmatrix}
1&-1\\ 
1& 1
\end{bmatrix}
\end{array} \nonumber
\end{eqnarray}
has a one defect only at the origin, where $\delta_{0}(0)=1$ and $\delta_{0}(x)=0 ~ (x\neq 0)$. Here, $|\omega -1|$ represents the strength of perturbation from the homogeneous time-evolution operator $U_1$. Moreover, when $\omega \neq \pm 1$, $C(0)$ is not unitary and $U_{\omega}$ becomes a non-unitary time-evolution operator. 
Thus, our evolution $U_\omega$ defines a non-unitary one defect quantum-walk model.

\section{Main results}

Let
\begin{align}
R_{\pm}(\omega) =\sqrt{\frac{\pm \omega(\omega-1)^2+\sqrt{\omega^2 ((\omega -1)^4 +(\omega^2-\omega +1)^2 )}}{4(\omega -\frac{1}{2})^2+1}} \label{202}
\end{align}
be functions defined for $\omega \in \mathbb{R}$. In the following, for a complex number $z\in \mathbb{C} \setminus \{0\}$ expressed in the polar form as $z=re^{i\theta}$ with $-\pi < \theta \leq \pi$, the square root $\sqrt{z}$ is defined by
\begin{align*}
    \sqrt{z} = z^{\frac{1}{2}}=r^{\frac{1}{2}}e^{\frac{i\theta}{2}}.
\end{align*}
Here, since $\omega^2 ((\omega -1)^4 +(\omega^2-\omega +1)^2) \geq 0$ and $\pm \omega(\omega-1)^2+\sqrt{\omega^2 ((\omega -1)^4 +(\omega^2-\omega +1)^2 )} > 0$, it follows that $R_{\pm}(\omega)> 0$.

The following is the main result of this paper.

\begin{thm} \label{thm31}
Suppose $\omega \in \mathbb{R}\backslash \{0,1\}$. 
\begin{itemize}
\item[(i)] $U_\omega$ has four eigenvalues which are expressed as
\begin{align}
\lambda_{1}=R_{-}(\omega)+iR_{+}(\omega),\quad
  \lambda_{2}=-\bar{\lambda_{1}}=-R_{-}(\omega)+iR_{+}(\omega),\nonumber \\
  \lambda_{3}=-\lambda_{1}=-R_{-}(\omega)-iR_{+}(\omega),\quad \lambda_{4}=\bar{\lambda_{1}}=R_{-}(\omega)-iR_{+}(\omega) . \nonumber
\end{align}
\item[(ii)]
When $\lambda = R_-(\omega) \pm i \mathrm{sgn}(\omega)R_+(\omega)$, the corresponding eigenvector $\Psi (x)$ is
\begin{eqnarray}
\Psi(x)=
\left\{\begin{array}{ll}
\begin{bmatrix}
\mp i z_{-}^{x}\Bigl(\frac{-\lambda +\lambda^{-1}+\sqrt{\lambda^2+\lambda^{-2}}}{\sqrt{2}} \Bigr) \\
\pm i z_{-}^{x-1}
\end{bmatrix} & (x\geq 1 ), \\\\
\begin{bmatrix}
    \mp i\Bigl(\frac{-\lambda +\lambda^{-1}+\sqrt{\lambda^2+\lambda^{-2}}}{\sqrt{2}} \Bigr) \\
    \frac{-\lambda +\lambda^{-1}+\sqrt{\lambda^2+\lambda^{-2}}}{\sqrt{2}}
\end{bmatrix} & (x=0), \\\\
\begin{bmatrix}
z_{+}^{-|x+1|} \\
z_{+}^{-|x|}\Bigl(\frac{-\lambda +\lambda^{-1}+\sqrt{\lambda^2+\lambda^{-2}}}{\sqrt{2}} \Bigr)
\end{bmatrix} & (x\leq -1)
\end{array}\right. \nonumber
\end{eqnarray}
up to a scalar multiple, where $z_{+} = \frac{\lambda +\lambda ^{-1}+\sqrt{\lambda^2 + \lambda^{-2}}}{\sqrt{2}}$ and $z_{-}=\frac{\lambda +\lambda ^{-1}-\sqrt{\lambda^2 + \lambda^{-2}}}{\sqrt{2}}$.
When $\lambda = -R_-(\omega) \pm i \mathrm{sgn}(\omega)R_+(\omega)$, the corresponding eigenvector $\Psi (x)$ is 
\begin{eqnarray}
\Psi(x)=
\left\{\begin{array}{ll}
\begin{bmatrix}
\mp i z_{+}^{x} \Bigl(\frac{-\lambda +\lambda^{-1}-\sqrt{\lambda^2+\lambda^{-2}}}{\sqrt{2}} \Bigr) \\
\pm i z_{+}^{x-1}
\end{bmatrix} & (x\geq 1 ), \\\\
\begin{bmatrix}
    \mp i \Bigl(\frac{-\lambda +\lambda^{-1}-\sqrt{\lambda^2+\lambda^{-2}}}{\sqrt{2}} \Bigr) \\
    \frac{-\lambda +\lambda^{-1}-\sqrt{\lambda^2+\lambda^{-2}}}{\sqrt{2}}
\end{bmatrix} & (x=0) \\\\
\begin{bmatrix}
z_{-}^{-|x+1|} \\
z_{-}^{-|x|}\Bigl(\frac{-\lambda +\lambda^{-1}-\sqrt{\lambda^2+\lambda^{-2}}}{\sqrt{2}} \Bigr)
\end{bmatrix} & (x\leq -1)
\end{array}\right. \nonumber
\end{eqnarray}
up to a scalar multiple.
In particular, $ {\rm dim ker} (U_{\omega}-\lambda_{j})=1 \quad (j=1,2,3,4). $
\end{itemize}
\end{thm}

\section{Proof of Main results}

In this section, we introduce the transfer matrix for our $U_\omega$, the main tool of our study. We note that $0$ is not an eigenvalue of $U_{\omega}$, since the inverse $U_{\omega}^{-1}=C^{-1}S^{-1}$ exists and is bounded for any $\omega \in \mathbb{R}\setminus \{ 0 \}$.\\

 Put $T_{\lambda \infty}=\begin{bmatrix}
    \sqrt{2}\lambda & 1\\
    1 & \sqrt{2}\lambda^{-1}
\end{bmatrix} \in SL(2,\mathbb{C})=\{M \in M(2,\mathbb{C})|\det M =1 \}$ for $\lambda \in \mathbb{C}\setminus \{ 0\}$and define the transfer matrix of $U_\omega$ as 
\begin{eqnarray}
T_{\lambda}(x)=
\left\{\begin{array}{ll}
T_{\lambda \infty} & (x= \pm 1, \pm 2, \cdots ), \\\\
\begin{bmatrix}
\frac{\sqrt{2}\lambda}{\omega} & 1 \\ 
1 & \frac{\sqrt{2}\omega}{\lambda}
\end{bmatrix} & (x=0). 
\end{array}\right. \nonumber
\end{eqnarray}
Then $x \mapsto T_{\lambda}(x)$ is the map from $\mathbb{Z}$ to $SL(2,\mathbb{C})$.

To solve the eigenvalue problem $U_\omega \Psi =\lambda \Psi$, we transform the vector $\begin{bmatrix}
    \Psi _{L}(x) \\
    \Psi _{R}(x)
\end{bmatrix}$
into $\begin{bmatrix}
    \Psi _{L}(x-1) \\
    \Psi _{R}(x)
\end{bmatrix}$
by the unitary operator 
    $J=L^{*}\oplus I_{\ell^2(\mathbb{Z})}$
    on $\: \mathcal{H} \simeq \ell^2 (\mathbb{Z})\oplus \ell^2 (\mathbb{Z})$
and rewrite the eigenvalue problem in term of the transfer matrix.

\begin{lemma}\label{lemma31}
Let $\lambda \in \mathbb{C}\setminus\{0 \}$ and $\Psi \in \mathcal{H}$. Then the following are equivalent.
\begin{itemize}
\item[(i)] $\Psi \in {\rm ker}(U_{\omega} -\lambda)$.
\item[(ii)] $\Psi$ satisfies
\begin{align}
    (L\oplus L)J\Psi =T_{\lambda}J\Psi \label{31},
\end{align} where $T_{\lambda}:\mathcal{H} \rightarrow \mathcal{H}$ is defined as $(T_{\lambda}\Phi)(x)=T_{\lambda}(x)\Phi(x)$ for $\Phi \in \mathcal{H}$.
\end{itemize}
\begin{proof}
We first prove that ($\mathrm{i})$ implies ($\mathrm{ii}$). Let $(U_{\omega}-\lambda)\Psi=0$, then

\begin{eqnarray}
\label{33}
\lambda \Psi_{L}(x-1) = \frac{\omega_{x}}{\sqrt{2}}(\Psi_{L}(x)-\Psi_{R}(x)),\\
\label{34}
\lambda \Psi_{R}(x+1) = \frac{\omega_{x}}{\sqrt{2}}(\Psi_{L}(x)+\Psi_{R}(x)),
\end{eqnarray}

for any $x\in \mathbb{Z}$, where $\omega_x=\omega $ when $x=0$, and $\omega_x=1$ when $x\neq0$. Subtracting Eq.(\ref{33}) from Eq.(\ref{34}), we have the equation 
\begin{eqnarray}
\label{35}
\Psi_{R}(x+1)= \Psi_{L}(x-1)+\frac{\sqrt{2}}{\lambda}\omega_{x}\Psi_{R}(x).
\end{eqnarray}
Hence, from Eq.(\ref{33}) and Eq.(\ref{35}), we obtain
\begin{align*}
((L \oplus L)J\Psi)(x) =&
\begin{bmatrix}
    \Psi_{L}(x) \\
    \Psi_{R}(x+1)
\end{bmatrix}\\
=&
\begin{bmatrix}
\frac{\sqrt{2}\lambda}{\omega_{x}}\Psi_{L}(x-1)+\Psi_{R}(x) \\
\Psi_{L}(x-1)+\frac{\sqrt{2}\omega_{x}}{\lambda}\Psi_{R}(x)
\end{bmatrix}\nonumber\\
=&\begin{bmatrix}
\frac{\sqrt{2}\lambda}{\omega_{x}} & 1 \\
1 & \frac{\sqrt{2}\omega_{x}}{\lambda}
\end{bmatrix}
\begin{bmatrix}
\Psi_{L}(x-1) \\
\Psi_{R}(x)
\end{bmatrix}\nonumber\\
=& (T_{\lambda}J\Psi)(x) \nonumber
\end{align*}
which leads to ($\mathrm{ii})$.\\

Next, we prove that (ii) implies (i). Assume $(L\oplus L)J\Psi = T_\lambda J\Psi$, then the equation
\begin{align*}
\begin{bmatrix}
    \Psi_{L}(x) \\
    \Psi_{R}(x+1)
\end{bmatrix}
=
((L \oplus L)J\Psi)(x) 
=
 (T_{\lambda}J\Psi)(x)
 =
\begin{bmatrix}
\frac{\sqrt{2}\lambda}{\omega_{x}}\Psi_{L}(x-1)+\Psi_{R}(x) \\
\Psi_{L}(x-1)+\frac{\sqrt{2}\omega_{x}}{\lambda}\Psi_{R}(x)
\end{bmatrix}
\end{align*}
holds. This equation leads to the Eq.(\ref{33}) and Eq.(\ref{34}).
Consequently, we have $\Psi \in \ker(U_{\omega} - \lambda)$.
\end{proof}
\end{lemma}
Here Eq.(\ref{31}) is equivalent to
\begin{align*}
    (J\Psi)(x+1)=T_{\lambda}(x)(J\Psi) (x), \quad x \in \mathbb{Z}.
\end{align*}
Hence, for $\Psi \in \ker (U_{\omega}-\lambda)$, multiplying $T_{\lambda}$ to $J\Psi$ is equivalent to applying the left shift to $J\Psi$. This is why we call $T_{\lambda}$ the transfer matrix. On the other hand, by Lemma 4.1 and the iteration of the above equation, $\Psi \in {\rm ker}(U_{\omega}-\lambda)$ holds if and only if $J\Psi$ satisfies
\begin{eqnarray}
(J\Psi)(x)=
\left\{\begin{array}{ll}
T_{\lambda \infty}^{x-1}T_{\lambda}(0)(J\Psi)(0) & (x\geq 1), \\\\
T_{\lambda \infty}^{x}(J\Psi)(0) & (x\leq -1).
\end{array}\right. \label{0318}
\end{eqnarray} 
To calculate $T_{\lambda \infty}^x$, we need to show the eigenvalues and eigenvectors of $T_{\lambda \infty}$.
\begin{lemma}\label{lemma32}
\begin{itemize}
\item[(i)] The eigenvalues of $T_{\lambda \infty}$ are
\begin{eqnarray}
\label{4210}
z_{+} = \frac{\lambda +\lambda ^{-1}+\sqrt{\lambda^2 + \lambda^{-2}}}{\sqrt{2}},\quad z_{-}=\frac{\lambda +\lambda ^{-1}-\sqrt{\lambda^2 + \lambda^{-2}}}{\sqrt{2}}. 
\end{eqnarray}
Here, $z_+$ and $z_-$ can be possibly the same.
\item[(ii)]The eigenvector of $T_{\lambda \infty}$ associated with $z_{\pm}$ is a constant multiple of  
\begin{align}
\chi_{\pm} = 
\begin{bmatrix}
   1 \\
   \frac{-\lambda +\lambda^{-1}\pm \sqrt{\lambda^2 + \lambda^{-2}}}{\sqrt{2}}
\end{bmatrix}. \nonumber
\end{align}
When $z_+$ and $z_-$ are the same, $\chi_+$ and $\chi_-$ are also the same, and the dimension of the eigenspace is one.
\end{itemize}

\begin{proof}
Solving the quadratic equation
$
\det(\mu -T_{\lambda \infty})
= \mu^2 -\big(\frac{\sqrt{2}}{\lambda} + \sqrt{2}\lambda \big) \mu +1 =0
$, we obtain ($\mathrm{i}$). One can check ($\mathrm{ii}$) by direct calculation.
\end{proof}
\end{lemma}

\begin{remark}
Noting $\det T_{\lambda \infty}=1$, we observe
\begin{equation}\label{eq4.2.1}
z_{+}=\frac{1}{z_{-}}.
\end{equation}
\end{remark}

To consider the convergence of $\|J\Psi\|$, it is important to know the value $|z_+|$. To do this, we define the following sets:
\begin{align*}
    \Sigma &= \Bigl\{\lambda =e^{i\theta} \colon \theta \in \Bigl[\frac{\pi}{4},\frac{3\pi}{4}\Bigr]\cup \Bigl[\frac{5\pi}{4},\frac{7\pi}{4}\Bigr] \Bigr\}, \\
    \Sigma_0 &=\{e^{i\frac{\pi}{4}}, e^{i\frac{3\pi}{4}}, e^{i\frac{5\pi}{4}}, e^{i\frac{7\pi}{4}} \}, \\
    \Xi_+ &=(\{z\in \mathbb{C} \colon \Re z>0 \} \cup \{it \colon t\in(-1,0)\cup (1,\infty)\} )\backslash \Sigma, \\
    \Xi_- &= (\{z\in \mathbb{C} \colon \Re z<0 \} \cup \{it \colon t\in(-\infty,-1)\cup (0,1)\}) \backslash \Sigma.
\end{align*}
Then, we have the next lemma. 
\begin{lemma}\label{lemma43}
\end{lemma}
\begin{itemize}
\item[(i)] $z_+=z_-$ holds if and only if $\lambda \in \Sigma_0$.

\item[(ii)] $|z_+|=1$ holds if and only if $\lambda \in \Sigma$.

\item[(iii)]$|z_+|>1$ holds if and only if $\lambda \in \Xi_{+}$, while $|z_+|<1$ holds if and only if $\lambda \in \Xi_{-}$.

\end{itemize}
We postpone the proof of Lemma \ref{lemma43} to Appendix.
\begin{lemma}\label{lemma33}
\begin{itemize}
\item[(i)]When $\lambda \in \Xi_{+}$, a non-zero solution $\Psi$ of Eq.(\ref{31}) in Lemma \ref{lemma31} exists in $\mathcal{H}$  if and only if $T_{\lambda}(0)\chi_{+}$ and $\chi_{-}$ are linearly dependent. In this case, the solution $\Psi$ is unique up to a constant factor.

\item[(ii)] When $\lambda \in \Xi_-$, a non-zero solution $\Psi$ of Eq.(\ref{31}) in Lemma \ref{lemma31} exists in $\mathcal{H}$ if and only if $T_{\lambda}(0)\chi_{-}$ and $\chi_{+}$ are linearly dependent. In this case, the solution $\Psi$ is unique up to a constant factor.
\item[(iii)]When $\lambda \in \Sigma$, ${\rm ker}(U_{\omega}-\lambda)=\{0\}$.
\end{itemize}
\end{lemma}

\begin{proof}
First, we assume $\lambda \notin \Sigma_0$ so that $z_+ \neq z_-$. Let $\Psi $ be a solution of Eq.(\ref{31}) or equivalently $\Psi \in \ker (U_{\omega}-\lambda)$. Since $\{ \chi_+ ,\chi_- \}$ is a basis of $\mathbb{C}^2$, there exist scalars $a_{\pm}, b_{\pm} \in \mathbb{C}$ such that 
\begin{align}
\label{4411}
    T_{\lambda}(0)(J\Psi)(0) =a_{+} \chi_{+} +a_{-}\chi_{-}, \quad (J\Psi)(0) =b_{+} \chi_{+} +b_{-}\chi_{-}.
\end{align}
We have from Eq.(\ref{0318})
\begin{align*}
    J\Psi (x) &= T_{\lambda \infty }^{x-1}T_{\lambda}(0)(J\Psi)(0)\nonumber \\
    &=z_{+}^{x-1} a_{+} \chi_{+} +z_{-} ^{x-1}a_{-}\chi_{-} \quad (x \geq 1)
\end{align*}
and
\begin{align*}
    J\Psi (x) &= T_{\lambda \infty }^{x}(J\Psi)(0)\nonumber \\
    &=z_{+}^{x} b_{+} \chi_{+} +z_{-} ^{x}b_{-}\chi_{-} \quad (x \leq 0).
\end{align*}
Hence, the norm of $\Psi$ is calculated by 
\begin{align*}
    \| \Psi \|^2 &= \| J\Psi \|^2 = \sum_{x\in \mathbb{Z}}\| J\Psi (x) \|^2 \\
    &=\sum_{x \geq 1} \| z_{+}^{x-1} a_{+} \chi_{+} +z_{-} ^{x-1}a_{-}\chi_{-} \|^2 + \sum_{x \leq 0} \| z_{+}^{x} b_{+} \chi_{+} +z_{-} ^{x}b_{-}\chi_{-} \|^2.
\end{align*}

(i) When $\lambda \in \Xi_+$, $|z_+|>1$ and $|z_-|<1$ by Lemma \ref{lemma43} and Eq. \eqref{eq4.2.1}. Therefore, the necessary condition for $\| \Psi \|$ to converge is $a_+=0$ and $b_-=0$. This is because, if $a_+ \neq 0$, then
\begin{align*}
\|z_{+}^{x-1}a_+\chi_+ + z_{-}^{x-1}a_-\chi_- \| \geq |z_{+}|^{x-1}|a_+|\|\chi_+ \| - |z_{-}|^{x-1}|a_-|\|\chi_- \|
\end{align*}
which implies that the expression diverges as $x \to \infty$. Since the norm does not tend to zero, the corresponding series cannot be convergent. An analogous conclusion holds for $b_-$ based on the second sum. This means, from Eq.(\ref{4411}),
\begin{align*}
    T_{\lambda}(0)b_+\chi_+=a_- \chi_-.
\end{align*}
Hence, if Eq.(\ref{31}) has a non-zero solution, $T_{\lambda}(0)\chi_+$ and $\chi_-$ are linearly dependent. Furthermore, the solution $\Psi$ of Eq.(\ref{31}) is unique up to a constant factor.  Conversely if $T_{\lambda}(0)\chi_+$ and $\chi_-$ are linearly dependent, the vector $\Phi$ defined by $J\Phi (0)=\chi_+$ and Eq.(\ref{0318}) is a non-zero solution of Eq.(\ref{31}).

(ii) It can be shown by performing calculations similar to those in (i).

(iii) Assume $\lambda \in \Sigma \backslash \Sigma_0$ and $\Psi$ is a non-zero solution of Eq.\eqref{31}. 
Noting $|z_{+}|=|z_{-}|=1$, for $x \ge 1$,
\begin{align*}
    \| (J\Psi)(x)\|^2
    &= \|z_{+}^{x-1} a_{+} \chi_{+} +z_{-} ^{x-1}a_{-}\chi_{-} \|^2 \\
    &= |a_{+}|^2 \|\chi_{+} \|^2 + |a_{-}|^2 \| \chi_{-} \|^2 + 2\Re \bar{z_{+}}^{x-1} \bar{a_{+}}z_{-}^{x-1} a_{-}\langle\chi_{+},\chi_{-}\rangle \\
    &\geq \| a_{+}\chi _{+} \|^2 + \|a_{-}\chi_{-} \|^2 -2|\langle a_{+}\chi_{+},a_{-}\chi_{-} \rangle| =: \epsilon .
    \end{align*}
Here, $\epsilon$ is independent of $x$.
By Schwartz's inequality, we obtain 
\begin{align*}
  \epsilon &\geq \| a_{+}\chi _{+} \|^2 + \|a_{-}\chi_{-} \|^2 -2\| a_{+}\chi_{+}\|\cdot \|a_{-}\chi_{-}\| \\
  &=(\| a_{+}\chi_{+}\|-\|a_{-}\chi_{-}\|)^2.
\end{align*}

Since $\Psi$ is not zero, $(a_+, a_-) \neq (0,0)$. Thus, the equality holds true only when $\chi_{+}$ and $\chi_{-}$ are linearly dependent. However, this contradicts to the fact that $\chi_{+}$ and $\chi_{-}$ are the eigenvectors of different eigenvalues.
Hence, $\epsilon$ is strictly positive and 
\begin{align*}
\inf _{x \geq 1} \| J\Psi (x) \|^2 \geq \epsilon > 0.
\end{align*}
Therefore,
\begin{align*}
    \|\Psi\|^2 = \|J\Psi \|^2 &= \sum _{x\in \mathbb{Z}} \|J\Psi (x) \|^2 \geq \sum _{x=1}^{N}\|J\Psi (x) \|^2 \geq N\epsilon \rightarrow  \infty \quad (N\rightarrow \infty).
\end{align*}

Thus, $\Psi \notin \mathcal{H}$, which indicates $\ker(U_{\omega}-\lambda) =\{0\}$.

Finally, we assume $\lambda \in \Sigma_0$ and $\Psi$ is a solution of Eq.\eqref{31} and therefore $\Psi$ satisfies Eq.\eqref{0318}. Considering the Jordan decomposition, there exists $\xi \in \mathbb{C}^2$ that satisfies
\begin{align*}
    T_{\lambda \infty}\xi = z_+ \xi +\chi_+.
\end{align*}
Since the vectors $\xi$ and $\chi_+$ form a basis, $T_{\lambda}(0)(J\Psi)(0)$ can be expressed as
\begin{align*}
    T_{\lambda}(0)(J\Psi)(0)=a\chi_+ +b\xi.
\end{align*}
Then, we have
\begin{align*}
    \sum_{x \geq 1}\|(J\Psi)(x) \|^2 &= \sum_{x \geq 1}\|T_{\lambda \infty}^{x-1}T_{\lambda}(0)(J\Psi)(0) \|^2 \\
    &= \sum_{x \geq 1} \|(az_{+}^{x-1}+b(x-1)z_{+}^{x-2})\chi_{+} +z_{+}^{x-1}b\xi \|^2. 
\end{align*}
When $b \neq 0$, the above expression diverges. The reason is that
\begin{align*}
\|(az_{+}^{x-1}+b(x-1)z_{+}^{x-2})\chi_{+} + z_{+}^{x-1}b\xi \| \geq |b|(x-1)\|\chi_+ \| - |a|\|\chi_+ \| - |b|\|\xi \|
\end{align*}
which diverges as $x \to \infty$. Thus, $b$ must be zero. Moreover, since $|z_+|=1$, 
\begin{align*}
   \sum_{x \geq 1}\|(J\Psi)(x) \|^2 \ge \sum_{x=1}^N\|(J\Psi)(x) \|^2 = N |a|^2\|\chi_+\|^2 \to \infty \quad (N\to \infty).
\end{align*}
Therefore, $a=0$. This means that Eq.\eqref{31} has no non-zero solution, and therefore $\ker(U_{\omega}-\lambda) =\{0\}$.
\end{proof}

\begin{lemma}\label{lemma34}
\begin{itemize}
\item[(i)]When $\lambda \in \Xi_{+}$, $T_{\lambda}(0)\chi _{+}$ and $\chi_{-}$ are linearly dependent if and only if $\frac{\omega}{\lambda}(\frac{-\lambda +\lambda^{-1}+\sqrt{\lambda^2+\lambda^{-2}}}{\sqrt{2}})=\frac{-1\pm i}{\sqrt{2}}$.
In this case, $T_\lambda(0)\chi_+=  \mp i \left(\frac{-\lambda +\lambda^{-1}+\sqrt{\lambda^2+\lambda^{-2}}}{\sqrt{2}}\right)\chi_-$.
\item[(ii)]When $\lambda \in \Xi _{-}$, $T_{\lambda}(0)\chi _{-}$ and $\chi_{+}$ are linearly dependent if and only if $\frac{\lambda}{\omega}(\frac{-\lambda +\lambda^{-1}+\sqrt{\lambda^2 +\lambda^{-2}}}{\sqrt{2}})=\frac{1\pm i}{\sqrt{2}}$.
In this case, $T_\lambda(0)\chi_-=  \mp i \left(\frac{-\lambda +\lambda^{-1}+\sqrt{\lambda^2+\lambda^{-2}}}{\sqrt{2}}\right) \chi_+$.
\end{itemize}
\end{lemma}

\begin{proof}
($\mathrm{i}$) $T_{\lambda}(0)\chi_{+}$ and $\chi_{-}$ are linearly dependent if and only if $\det \begin{bmatrix}
T_{\lambda}(0)\chi_+& \chi_- 
\end{bmatrix}
=0$. Since $\left(\frac{-\lambda +\lambda^{-1}+\sqrt{\lambda^2+\lambda^{-2}}}{\sqrt{2}}\right)^{-1}=-\left(\frac{-\lambda +\lambda^{-1}-\sqrt{\lambda^2+\lambda^{-2}}}{\sqrt{2}}\right)$, we have
\begin{gather*}
\det \begin{bmatrix}
T_{\lambda}(0)\chi_+& \chi_- 
\end{bmatrix}= \det
\begin{bmatrix}
   \frac{\sqrt{2}\lambda}{\omega} +\frac{-\lambda +\lambda^{-1}+\sqrt{\lambda^2+\lambda^{-2}}}{\sqrt{2}} & 1 \\
   1 + \frac{\sqrt{2}\omega}{\lambda}(\frac{-\lambda +\lambda^{-1}+\sqrt{\lambda^2+\lambda^{-2}}}{\sqrt{2}}) & \frac{-\lambda +\lambda^{-1}-\sqrt{\lambda^2+\lambda^{-2}}}{\sqrt{2}}
\end{bmatrix}\\
=
-2-\sqrt{2}\left( \frac{\omega}{\lambda}\Bigl(\frac{-\lambda +\lambda^{-1}+\sqrt{\lambda^2+\lambda^{-2}}}{\sqrt{2}}\Bigr)\right)-\sqrt{2}\left(\frac{\omega}{\lambda}\Bigl(\frac{-\lambda +\lambda^{-1}+\sqrt{\lambda^2+\lambda^{-2}}}{\sqrt{2}}\Bigr)\right)^{-1}.\nonumber 
\end{gather*}
Putting $\Lambda_{1} := \frac{\omega}{\lambda}\Bigl(\frac{-\lambda +\lambda^{-1}+\sqrt{\lambda^2+\lambda^{-2}}}{\sqrt{2}}\Bigr)$, $\det \begin{bmatrix}
T_{\lambda}(0)\chi_+& \chi_- 
\end{bmatrix}
=0$ is equivalent to
\begin{eqnarray*}
\Lambda_{1}+\Lambda_{1}^{-1}+\sqrt{2}=0.
\end{eqnarray*}
Solving this equation for $\Lambda_1$, we obtain
\begin{eqnarray*}
\Lambda_{1} = \frac{-1\pm i}{\sqrt{2}}
\end{eqnarray*}
which follows
\begin{eqnarray*}
\frac{\omega}{\lambda}\left(\frac{-\lambda +\lambda^{-1}+\sqrt{\lambda^2+\lambda^{-2}}}{\sqrt{2}}\right)=\frac{-1\pm i}{\sqrt{2}}.
\end{eqnarray*}

When $T_{\lambda}(0)\chi_{+}$ and $\chi_{-}$ are linearly dependent, there exists a nonzero scalar $\gamma \in \mathbb{C}$ such that 
\begin{eqnarray}
\label{320}
T_{\lambda}(0)\chi_{+} = \gamma \chi_{-}. 
\end{eqnarray}
Since $\frac{\lambda}{\omega}=\frac{-1\mp i}{\sqrt{2}}\Bigl( \frac{-\lambda +\lambda^{-1}+\sqrt{\lambda^2+\lambda^{-2}}}{\sqrt{2}} \Bigr)$,
\begin{eqnarray}
T_{\lambda}(0)\chi_{+}=
\begin{bmatrix}
   \mp i\Bigl(\frac{-\lambda +\lambda^{-1}+\sqrt{\lambda^2+\lambda^{-2}}}{\sqrt{2}}\Bigr) \\
   \pm i
\end{bmatrix}. \nonumber 
\end{eqnarray}
Hence, Eq.(\ref{320}) can be computed as
\begin{eqnarray}
\begin{bmatrix}
   \mp i\Bigl(\frac{-\lambda +\lambda^{-1}+\sqrt{\lambda^2+\lambda^{-2}}}{\sqrt{2}}\Bigr) \\
   \pm i
\end{bmatrix} \nonumber 
=\gamma 
\begin{bmatrix}
   1 \\
    \Bigl(\frac{-\lambda +\lambda^{-1}-\sqrt{\lambda^2+\lambda^{-2}}}{\sqrt{2}}
    \Bigr)
\end{bmatrix},
\end{eqnarray}
and therefore,
\begin{align*}
\gamma  = \mp i\left(\frac{-\lambda +\lambda^{-1}+\sqrt{\lambda^2+\lambda^{-2}}}{\sqrt{2}}\right).
\end{align*}

($\mathrm{ii}$) $T_{\lambda}(0)\chi_{-}$ and $\chi_{+}$ are linearly dependent if and only if $\det \begin{bmatrix}
T_{\lambda}(0)\chi_-& \chi_+ 
\end{bmatrix}
=0$. Proceeding in the same manner as above,
\begin{gather*}
\det \begin{bmatrix}
T_{\lambda}(0)\chi_-& \chi_+ 
\end{bmatrix}= \det
\begin{bmatrix}
   \frac{\sqrt{2}\lambda}{\omega} +\frac{-\lambda +\lambda^{-1}-\sqrt{\lambda^2+\lambda^{-2}}}{\sqrt{2}} & 1 \\
   1 + \frac{\sqrt{2}\omega}{\lambda}\Bigl(\frac{-\lambda +\lambda^{-1}-\sqrt{\lambda^2+\lambda^{-2}}}{\sqrt{2}}\Bigr) & \frac{-\lambda +\lambda^{-1}+\sqrt{\lambda^2+\lambda^{-2}}}{\sqrt{2}}
\end{bmatrix} \\
=-2+\sqrt{2}\left( \frac{\lambda}{\omega}\Bigl(\frac{-\lambda +\lambda^{-1}+\sqrt{\lambda^2+\lambda^{-2}}}{\sqrt{2}}\Bigr) \right)+\sqrt{2}\left(\frac{\lambda}{\omega}\Bigl(\frac{-\lambda +\lambda^{-1}+\sqrt{\lambda^2+\lambda^{-2}}}{\sqrt{2}}\Bigr)\right)^{-1}.
\end{gather*}
Putting $\Lambda_{2} := \frac{\lambda}{\omega}\Bigl(\frac{-\lambda +\lambda^{-1}+\sqrt{\lambda^2+\lambda^{-2}}}{\sqrt{2}}\Bigr)$, $\det \begin{bmatrix}
T_{\lambda}(0)\chi_-& \chi_+ 
\end{bmatrix}=0$ is equivalent to

\begin{eqnarray*}
\Lambda_{2}+\Lambda_{2}^{-1}-\sqrt{2}=0.
\end{eqnarray*}
Solving this equation for $\Lambda_2$, we obtain
\begin{eqnarray*}
\Lambda_{2} = \frac{1\pm i}{\sqrt{2}}
\end{eqnarray*}
which follows
\begin{eqnarray*}
\frac{\lambda}{\omega}\left(\frac{-\lambda +\lambda^{-1}+\sqrt{\lambda^2+\lambda^{-2}}}{\sqrt{2}}\right) = \frac{1\pm i}{\sqrt{2}}.
\end{eqnarray*}

When $T_{\lambda}(0)\chi_{-}$ and $\chi_{+}$ are linearly dependent, there exists a nonzero scalar $\gamma' \in \mathbb{C}$ such that
\begin{eqnarray}
\label{325}
T_{\lambda}(0)\chi_{-} = \gamma' \chi_{+}. 
\end{eqnarray}
Since
\begin{eqnarray}
T_{\lambda}(0)\chi_{-}=
\begin{bmatrix}
   \mp i\Bigl(\frac{-\lambda +\lambda^{-1}-\sqrt{\lambda^2+\lambda^{-2}}}{\sqrt{2}}\Bigr) \\
   \pm i
\end{bmatrix}, \nonumber 
\end{eqnarray}
Eq.(\ref{325}) can be computed as
\begin{eqnarray}
\begin{bmatrix}
   \mp i\Bigl(\frac{-\lambda +\lambda^{-1}-\sqrt{\lambda^2+\lambda^{-2}}}{\sqrt{2}}\Bigr) \\
   \pm i
\end{bmatrix} \nonumber 
=\gamma' 
\begin{bmatrix}
   1 \\
    \Bigl(\frac{-\lambda +\lambda^{-1}+\sqrt{\lambda^2+\lambda^{-2}}}{\sqrt{2}}
    \Bigr)
\end{bmatrix},
\end{eqnarray}
and therefore,
\begin{align*}
\gamma'  = \mp i\left(\frac{-\lambda +\lambda^{-1}-\sqrt{\lambda^2+\lambda^{-2}}}{\sqrt{2}}\right).
\end{align*}
\end{proof}

The following standard result on complex numbers will be used below.
\begin{lemma}\label{lemma46}
For an arbitrary complex number $a+bi \ (a,b\in \mathbb{R})$, its square root $\sqrt{a+bi}$ is express as 

\begin{eqnarray}
\sqrt{a+bi} = \left\{
\begin{array}{ll}
\sqrt{\frac{a+\sqrt{a^2+b^2}}{2}}+i\sqrt{\frac{-a+\sqrt{a^2+b^2}}{2}} &  b\geq0,\\
\sqrt{\frac{a+\sqrt{a^2+b^2}}{2}}-i\sqrt{\frac{-a+\sqrt{a^2+b^2}}{2}} & b<0.
\end{array}
\right. \nonumber 
\end{eqnarray}
\end{lemma}

Now, we are at the position to cralify the eigenvalues of $U$. 
Remark that $R_{\pm}(\omega)$ are the values defined in Eq.(\ref{202}).

\begin{lemma}\label{lemma47}
\begin{itemize}
\item[(i)]When $\lambda \in \Xi_{+}$, $\frac{\omega}{\lambda} \left( \frac{-\lambda + \lambda^{-1} + \sqrt{\lambda^2 + \lambda^{-2}}}{\sqrt{2}} \right) = \frac{-1 \pm i}{\sqrt{2}}$ if and only if $\lambda = R_{-}(\omega) \mp i\mathrm{sgn}(\omega) R_{+}(\omega)$.
\item[(ii)]When $\lambda \in \Xi _{-}$, $\frac{\lambda}{\omega}(\frac{-\lambda +\lambda^{-1}+\sqrt{\lambda^2 +\lambda^{-2}}}{\sqrt{2}})=\frac{1\pm i}{\sqrt{2}}$ if and only if $\lambda = -R_{-}(\omega) \pm i \mathrm{sgn}(\omega) R_{+}(\omega)$.
\end{itemize}
\end{lemma}
\begin{proof}
($\mathrm{i}$) We solve the following equation with respect to $\lambda$:
\begin{eqnarray*}
\frac{\omega}{\lambda} \left( \frac{-\lambda + \lambda^{-1} + \sqrt{\lambda^2 + \lambda^{-2}}}{\sqrt{2}} \right) = \frac{-1 \pm i}{\sqrt{2}}. 
\end{eqnarray*}
By computing and simplifying this equation, we get
\begin{align}
    \label{413}
    \sqrt{\lambda^2+\lambda^{-2}}=(-1\pm i)\cdot \frac{\lambda}{\omega}+\lambda-\lambda^{-1}.
\end{align}
Squaring both sides of this equation and simplifying, we obtain
\begin{align*}
\lambda^2 &= \frac{(\omega^2 - \omega \pm i\omega)(-\omega \mp i(\omega -1))}{2\omega^2 -2\omega +1} \\
&= \frac{\omega (-\omega ^2 +2\omega -1) \mp i\omega (\omega ^2 -\omega +1)}{2\omega ^2 -2\omega +1}. 
 \end{align*}
We divide the problem into two cases. We first consider the case of 
\begin{align*}
\lambda ^2= \frac{\omega (-\omega ^2 +2\omega -1) - i\omega (\omega ^2 -\omega +1)}{2\omega ^2 -2\omega +1}. 
 \end{align*}
Since $\omega^2-\omega+1$ is positive for any $\omega \in \mathbb{R} \backslash \{0\}$, we observe that
\begin{eqnarray*}
\Im \lambda^2 = \left\{
\begin{array}{ll}
{\rm positive} &  {\rm if} \quad  \omega <0,\\
{\rm negative} & {\rm if} \quad \omega >0.
\end{array}
\right. 
\end{eqnarray*}
By Lemma \ref{lemma46} and the fact that $\Re \lambda \geq 0$,
\begin{eqnarray*}
\lambda = \left\{
\begin{array}{ll}
R_{-}(\omega)-iR_{+}(\omega) &   \omega >0,\\
R_{-}(\omega)+iR_{+}(\omega) &  \omega <0.
\end{array}
\right.
\end{eqnarray*}
The obtained expression for $\lambda$ can be written in the following form using $\mathrm{sgn}(\omega)$:
\begin{align*}
    \lambda = R_{-}(\omega)-i\mathrm{sgn}(\omega)R_{+}(\omega).
\end{align*}

 We next consider the case of 
\begin{align*}
\lambda ^2= \frac{\omega (-\omega ^2 +2\omega -1) + i\omega (\omega ^2 -\omega +1)}{2\omega ^2 -2\omega +1}, 
 \end{align*}
where 
\begin{eqnarray*}
\Im \lambda^2 = \left\{
\begin{array}{ll}
{\rm positive} &  {\rm if} \quad  \omega >0,\\
{\rm negative} & {\rm if} \quad \omega <0.
\end{array}
\right.
\end{eqnarray*}
Similarly as above, we obtain
\begin{align*}
    \lambda = R_{-}(\omega)+i\mathrm{sgn}(\omega)R_{+}(\omega).
\end{align*}

We need to verify whether the obtained values of $\lambda$ realy satisfy Eq.(\ref{413}), 
because we square Eq.\eqref{413} to solve the equation. Since $\Re \sqrt{\lambda^2 +\lambda^{-2}} \ge 0$, it is enough to see that $\Re \Bigl( (-1\pm i) \cdot \frac{\lambda}{\omega}+\lambda -\lambda^{-1} \Bigr) > 0$ for $\lambda = R_-(\omega) \mp i\mathrm{sgn}(\omega) R_+(\omega)$. 
Substituting $\lambda$ yields 
\begin{align*}
    (-1 \pm i) \cdot \frac{\lambda}{\omega} +\lambda -\lambda^{-1} &= -\frac{R_-(\omega)\mp i\mathrm{sgn}(\omega)R_+(\omega)}{\omega} + \frac{\pm iR_-(\omega) + \mathrm{sgn}(\omega)R_+(\omega)}{\omega}\\
    &+R_-(\omega)\mp i\mathrm{sgn}(\omega)R_+(\omega)-\frac{R_-(\omega)\pm i\mathrm{sgn}(\omega)R_+(\omega)}{R_-(\omega)^2 +R_+(\omega)^2}.
\end{align*}
By extracting the real part from both sides, we have
\begin{align*}
    \Re \Bigl( (-1\pm i) \cdot \frac{\lambda}{\omega} +\lambda -\lambda^{-1} \Bigr) =-\frac{R_-(\omega)}{\omega}+\frac{\mathrm{sgn}(\omega)R_+(\omega)}{\omega}+R_-(\omega) -\frac{R_-(\omega)}{R_-(\omega)^2+R_+(\omega)^2}.
\end{align*}
Considering $R_-(\omega) >0$, we only need to show
\begin{align}
    -\frac{1}{\omega}+ \mathrm{sgn}(\omega)\frac{1}{\omega}\frac{R_+(\omega)}{R_-(\omega)}+1 - \frac{1}{R_-(\omega)^2+R_+(\omega)^2} > 0. \label{5236}
\end{align}
Since $(\omega -1)^4+(\omega^2-\omega+1)^2=(2\omega^2-2\omega+1)(\omega^2-2\omega+2)$,
\begin{align*}
    \frac{R_+(\omega)}{R_-(\omega)} = \frac{\mathrm{sgn}(\omega)(\omega -1)^2 +\sqrt{(2\omega ^2 -2\omega +1)(\omega ^2 -2\omega +2)}}{\omega^2 -\omega +1}.
\end{align*}
Note that $2\omega^2 - 2\omega + 1 > 0$ and $\omega^2 - 2\omega + 2 > 0$ for any $\omega \in \mathbb{R} \setminus \{0\}$. Then, we find
\begin{align*}
R_-(\omega)^2+R_+(\omega)^2 &= \mathrm{sgn}(\omega) \omega \sqrt{\frac{\omega^2-2\omega +2}{2\omega^2 -2\omega +1}},
\end{align*}
and hence,
\begin{align*}
    \frac{1}{R_-(\omega)^2+R_+(\omega)^2} =\mathrm{sgn}(\omega)\frac{1}{\omega}\sqrt{\frac{2\omega^2 -2\omega +1}{\omega^2 -2\omega +2}}.
\end{align*}
Substituting these results into Eq.(\ref{5236}) and
multiplying by $\mathrm{sgn}(\omega) \omega(\omega^2 - \omega + 1)\sqrt{\omega^2-2\omega +2}$, we obtain that Eq.\eqref{5236} is equivalent to
\begin{align}
\label{4716}
(\omega-1)(\mathrm{sgn}(\omega)\omega^2 \sqrt{\omega^2-2\omega +2} -\sqrt{2\omega^2-2\omega+1})>0.
\end{align}

When $\omega > 0$, the equation
\begin{align*}
    \left(\omega^2 \sqrt{\omega^2 - 2\omega + 2}\right)^2 - \left(\sqrt{2\omega^2 - 2\omega + 1}\right)^2 =(\omega-1)(\omega^2+1)(\omega^2-\omega+1),
\end{align*}
shows that $\omega^2 \sqrt{\omega^2-2\omega +2} -\sqrt{2\omega^2-2\omega+1}$
is positive when $\omega>1$ and is negative when $0<\omega<1$.
Therefore, Eq.(\ref{4716}) is true.
When $\omega < 0$, it is clear that Eq.(\ref{4716}) holds.

($\mathrm{ii}$) We solve the equation
\begin{eqnarray*}
\frac{\lambda}{\omega}\Bigl(\frac{-\lambda +\lambda^{-1}+\sqrt{\lambda^2 +\lambda^{-2}}}{\sqrt{2}}\Bigr)=\frac{1\pm i}{\sqrt{2}}
\end{eqnarray*}
with respect to $\lambda$.
Rewriting the equation yields the following expression:
\begin{align}
    \label{416}
    \sqrt{\lambda^2+\lambda^{-2}}=(1\pm i)\cdot \frac{\omega}{\lambda}+\lambda-\lambda^{-1}.
\end{align}
Squaring both sides of this equation and simplifying, we obtain
\begin{align}
\lambda^2 &= \frac{(\pm i\omega^2 - \omega \mp i\omega)(1-\omega \pm i\omega )}{2\omega^2 -2\omega +1} \nonumber \\
&= \frac{\omega (-\omega ^2 +2\omega -1) \mp i\omega (\omega ^2 -\omega +1)}{2\omega ^2 -2\omega +1}. \nonumber
 \end{align}
We divide the problem into two case. We first consider the case of 
\begin{align*}
\lambda ^2= \frac{\omega (-\omega ^2 +2\omega -1) - i\omega (\omega ^2 -\omega +1)}{2\omega ^2 -2\omega +1},
 \end{align*}
where
\begin{eqnarray*}
\Im \lambda^2 = \left\{
\begin{array}{ll}
{\rm positive} &  {\rm if} \quad  \omega <0,\\
{\rm negative} & {\rm if} \quad \omega >0.
\end{array}
\right. 
\end{eqnarray*}
By Lemma \ref{lemma46} and the fact that $\Re \lambda \leq 0$,
\begin{eqnarray*}
\lambda = \left\{
\begin{array}{ll}
-R_{-}(\omega)+iR_{+}(\omega) &   \omega >0,\\
-R_{-}(\omega)-iR_{+}(\omega) &  \omega <0.
\end{array}
\right.
\end{eqnarray*}
It follows from the above that $\lambda$ can be written as
\begin{align*}
    \lambda = -R_{-}(\omega)+i\mathrm{sgn}(\omega)R_{+}(\omega).
\end{align*}

We next consider the case of 
\begin{align*}
\lambda ^2= \frac{\omega (-\omega ^2 +2\omega -1) + i\omega (\omega ^2 -\omega +1)}{2\omega ^2 -2\omega +1}, 
 \end{align*}
where 
\begin{eqnarray*}
\Im \lambda^2 = \left\{
\begin{array}{ll}
{\rm positive} &  {\rm if} \quad  \omega >0,\\
{\rm negative} & {\rm if} \quad \omega <0.
\end{array}
\right.
\end{eqnarray*}
Similarly as above, we obtain
\begin{align*}
    \lambda = -R_{-}(\omega)-i\mathrm{sgn}(\omega)R_{+}(\omega).
\end{align*}
We now verify that the obtained values of $\lambda$ realy satisfy Eq.(\ref{416}). Since $\Re \sqrt{\lambda^2 +\lambda^{-2}} \ge 0$, it is enough to show that $\Re \Bigl( (1\pm i) \cdot \frac{\lambda}{\omega}+\lambda -\lambda^{-1} \Bigr) > 0$ for $\lambda = -R_-(\omega) \pm i\mathrm{sgn}(\omega) R_+(\omega)$.
Substituting $\lambda$ yields
\begin{align*}
    (1 \pm i) \cdot \frac{\omega}{\lambda} +\lambda -\lambda^{-1} =& -R_{-}(\omega)\pm i\mathrm{sgn}(\omega)R_{+}(\omega)+(-1+\omega)\cdot \frac{-R_{-}(\omega)\mp i\mathrm{sgn}(\omega)R_{+}(\omega)}{R_-(\omega)^2+R_+(\omega)^2}\\
    &+i\omega\frac{\mp R_{-}(\omega)-i\mathrm{sgn}(\omega)R_{+}(\omega)}{R_-(\omega)^2+R_+(\omega)^2}.
\end{align*}
By extracting the real part from both sides, we get
\begin{align*}
    &\Re \Bigl( (1\pm i) \cdot \frac{\lambda}{\omega} +\lambda -\lambda^{-1} \Bigr) =\\
    &-R_-(\omega)-(-1+\omega)\cdot \frac{R_-(\omega)}{R_-(\omega)^2+R_+(\omega)^2}+\mathrm{sgn}(\omega)\omega\cdot  \frac{R_+(\omega)}{R_-(\omega)^2+R_+(\omega)^2}.
\end{align*}
Since $R_-(\omega)>0$, we only need to prove
\begin{align}
    -1-(-1+\omega)\cdot \frac{1}{R_-(\omega)^2+R_+(\omega)^2}+\mathrm{sgn}(\omega)\omega \cdot \frac{R_+(\omega)}{R_-(\omega)}\cdot  \frac{1}{R_-(\omega)^2+R_+(\omega)^2} > 0. \label{4723}
\end{align}
Substituting the expressions for $\frac{R_+(\omega)}{R_-(\omega)}$ and $\frac{1}{R_-(\omega)^2 + R_+(\omega)^2}$ computed in (i) into Eq.(\ref{4723}) and multiplying by 
$\mathrm{sgn}(\omega)\omega(\omega^2 - \omega + 1)\sqrt{\omega^2-2\omega+2}$,
we obtain that Eq.\eqref{4723} is equivalent to 
\begin{align*}
(\omega-1)(\mathrm{sgn}(\omega)\omega^2 \sqrt{\omega^2-2\omega +2} -\sqrt{2\omega^2-2\omega+1})>0.
\end{align*}
The above inequality coincides with Eq.(\ref{4716}). Therefore, we have $\Re \Bigl( (1 \pm i) \cdot \frac{\lambda}{\omega} +\lambda -\lambda^{-1} \Bigr) > 0$.
\end{proof}

From the above lemmas we conclude Theorem \ref{thm31}(i), and it is found that there exist four eigenvalues.  We now proceed to determine the corresponding eigenvectors.

\begin{lemma}\label{lemma48}
For an eigenvalue $\lambda$ of $U_{\omega}$, the corresponding eigenvector is determined as follows:

\begin{itemize}
\item[(i)]If $\lambda = R_-(\omega) \pm i \mathrm{sgn}(\omega)R_+(\omega)$, then the corresponding eigenvector is
\begin{eqnarray}
\Psi(x)=
\left\{\begin{array}{ll}
\begin{bmatrix}
\mp i z_{-}^{x}\Bigl(\frac{-\lambda +\lambda^{-1}+\sqrt{\lambda^2+\lambda^{-2}}}{\sqrt{2}} \Bigr) \\
\pm i z_{-}^{x-1}
\end{bmatrix} & (x\geq 1 ), \\\\
\begin{bmatrix}
    \mp i\Bigl(\frac{-\lambda +\lambda^{-1}+\sqrt{\lambda^2+\lambda^{-2}}}{\sqrt{2}} \Bigr) \\
    \frac{-\lambda +\lambda^{-1}+\sqrt{\lambda^2+\lambda^{-2}}}{\sqrt{2}}
\end{bmatrix} & (x=0), \\\\
\begin{bmatrix}
z_{+}^{-|x+1|} \\
z_{+}^{-|x|}\Bigl(\frac{-\lambda +\lambda^{-1}+\sqrt{\lambda^2+\lambda^{-2}}}{\sqrt{2}} \Bigr)
\end{bmatrix} & (x\leq -1)
\end{array}\right. \nonumber
\end{eqnarray}
up to a scalar multiple.

\item[(ii)]If $\lambda = -R_-(\omega) \pm i \mathrm{sgn}(\omega)R_+(\omega)$, then the corresponding eigenvector is
\begin{eqnarray}
\Psi(x)=
\left\{\begin{array}{ll}
\begin{bmatrix}
\mp i z_{+}^{x} \Bigl(\frac{-\lambda +\lambda^{-1}-\sqrt{\lambda^2+\lambda^{-2}}}{\sqrt{2}} \Bigr) \\
\pm i z_{+}^{x-1}
\end{bmatrix} & (x\geq 1 ), \\\\
\begin{bmatrix}
    \mp i \Bigl(\frac{-\lambda +\lambda^{-1}-\sqrt{\lambda^2+\lambda^{-2}}}{\sqrt{2}} \Bigr) \\
    \frac{-\lambda +\lambda^{-1}-\sqrt{\lambda^2+\lambda^{-2}}}{\sqrt{2}}
\end{bmatrix} & (x=0), \\\\
\begin{bmatrix}
z_{-}^{-|x+1|} \\
z_{-}^{-|x|}\Bigl(\frac{-\lambda +\lambda^{-1}-\sqrt{\lambda^2+\lambda^{-2}}}{\sqrt{2}} \Bigr)
\end{bmatrix} & (x\leq -1)
\end{array}\right. \nonumber
\end{eqnarray}
up to a scalar multiple.
\end{itemize}
\end{lemma}
\begin{proof}
($\mathrm{i}$) From Lemma~\ref{lemma31}, if $\lambda$ is an eigenvalue, then the corresponding eigenvector $\Psi$ satisfies Eq.(\ref{0318}) up to a scalar multiple. Furthermore, from Lemma~\ref{lemma33} and Lemma~\ref{lemma34}, we have $T_{\lambda}(0)\chi_+ = \gamma \chi_-$ and $J\Psi(0) = \chi_+$, where $\gamma =\mp i \Bigl( \frac{-\lambda +\lambda^{-1}+\sqrt{\lambda^2+\lambda^{-2}}}{\sqrt{2}} \Bigr)$. Substituting these into Eq.(\ref{0318}), we obtain
\begin{eqnarray}
J\Psi(x)=
\left\{\begin{array}{ll}
\begin{bmatrix}
\mp iz_{-}^{x-1}\Bigl(\frac{-\lambda +\lambda^{-1}+\sqrt{\lambda^2+\lambda^{-2}}}{\sqrt{2}} \Bigr) \\
\pm iz_{-}^{x-1}
\end{bmatrix} & (x\geq 1 ), \\\\
\begin{bmatrix}
    1 \\
    \frac{-\lambda +\lambda^{-1}+\sqrt{\lambda^2+\lambda^{-2}}}{\sqrt{2}}
\end{bmatrix} & (x=0), \\\\
\begin{bmatrix}
z_{+}^{-|x|} \\
z_{+}^{-|x|}\Bigl(\frac{-\lambda +\lambda^{-1}+\sqrt{\lambda^2+\lambda^{-2}}}{\sqrt{2}} \Bigr)
\end{bmatrix} & (x\leq -1),
\end{array}\right. \nonumber
\end{eqnarray}
where $z_+$ and $z_-$ are eigenvalues of $T_{\lambda \infty }$ defined in Eq.(\ref{4210}). Since the operator $J$ is given by $J=L^{*}\oplus I_{\ell^2(\mathbb{Z})}$, its adjoint is $J^*=L \oplus I_{\ell^2(\mathbb{Z})}$. Therefore, computing $\Psi(x)=(J^{-1}J\Psi)(x)$, we get
\begin{eqnarray}
\Psi(x)=
\left\{\begin{array}{ll}
\begin{bmatrix}
\mp i z_{-}^{x}\Bigl(\frac{-\lambda +\lambda^{-1}+\sqrt{\lambda^2+\lambda^{-2}}}{\sqrt{2}} \Bigr) \\
\pm i z_{-}^{x-1}
\end{bmatrix} & (x\geq 1 ), \\\\
\begin{bmatrix}
    \mp i\Bigl(\frac{-\lambda +\lambda^{-1}+\sqrt{\lambda^2+\lambda^{-2}}}{\sqrt{2}} \Bigr) \\
    \frac{-\lambda +\lambda^{-1}+\sqrt{\lambda^2+\lambda^{-2}}}{\sqrt{2}}
\end{bmatrix} & (x=0), \\\\
\begin{bmatrix}
z_{+}^{-|x+1|} \\
z_{+}^{-|x|}\Bigl(\frac{-\lambda +\lambda^{-1}+\sqrt{\lambda^2+\lambda^{-2}}}{\sqrt{2}} \Bigr)
\end{bmatrix} & (x\leq -1).
\end{array}\right. \nonumber
\end{eqnarray}

($\mathrm{ii}$) Similar to (i),  the corresponding eigenvector $\Psi$ to $\lambda$ satisfies $T_{\lambda}(0)\chi_+- = \gamma' \chi_+$ and $J\Psi(0) = \chi_-$, where $\gamma =\mp i \Bigl( \frac{-\lambda +\lambda^{-1}-\sqrt{\lambda^2+\lambda^{-2}}}{\sqrt{2}} \Bigr)$. Substituting these into Eq.(\ref{0318}), we obtain
\begin{eqnarray}
J\Psi(x)=
\left\{\begin{array}{ll}
\begin{bmatrix}
\mp i z_{+}^{x-1}\Bigl(\frac{-\lambda +\lambda^{-1}-\sqrt{\lambda^2+\lambda^{-2}}}{\sqrt{2}} \Bigr) \\
\pm iz_{+}^{x-1}
\end{bmatrix} & (x\geq 1 ), \\\\
\begin{bmatrix}
    1 \\
    \Bigl(\frac{-\lambda +\lambda^{-1}-\sqrt{\lambda^2+\lambda^{-2}}}{\sqrt{2}} \Bigr)
\end{bmatrix} & (x=0), \\\\
\begin{bmatrix}
z_{-}^{-|x|} \\
z_{-}^{-|x|}\Bigl(\frac{-\lambda +\lambda^{-1}-\sqrt{\lambda^2+\lambda^{-2}}}{\sqrt{2}} \Bigr)
\end{bmatrix} & (x\leq -1).
\end{array}\right. \nonumber
\end{eqnarray}
Therefore, we get
\begin{eqnarray}
\Psi(x)=
\left\{\begin{array}{ll}
\begin{bmatrix}
\mp i z_{+}^{x} \Bigl(\frac{-\lambda +\lambda^{-1}-\sqrt{\lambda^2+\lambda^{-2}}}{\sqrt{2}} \Bigr) \\
\pm i z_{+}^{x-1}
\end{bmatrix} & (x\geq 1 ), \\\\
\begin{bmatrix}
    \mp i \Bigl(\frac{-\lambda +\lambda^{-1}-\sqrt{\lambda^2+\lambda^{-2}}}{\sqrt{2}} \Bigr) \\
    \frac{-\lambda +\lambda^{-1}-\sqrt{\lambda^2+\lambda^{-2}}}{\sqrt{2}}
\end{bmatrix} & (x=0), \\\\
\begin{bmatrix}
z_{-}^{-|x+1|} \\
z_{-}^{-|x|}\Bigl(\frac{-\lambda +\lambda^{-1}-\sqrt{\lambda^2+\lambda^{-2}}}{\sqrt{2}} \Bigr)
\end{bmatrix} & (x\leq -1).
\end{array}\right. \nonumber
\end{eqnarray}
\end{proof}

From the above lemma, we conclude Theorem \ref{thm31}(ii).

\begin{remark}\label{remark3}
It is well-known that the spectrum and the essential spectrum of the homogeneous time-evolution operator $U_1$ are both $\Sigma=\left\{e^{i\theta}:\theta\in\left[\frac{\pi}{4},\frac{3\pi}{4}\right]\cup\left[\frac{5\pi}{4},\frac{7\pi}{4}\right]\right\}$.
Since $U_\omega - U_1$ for $\omega\in {\mathbb R}\backslash\{0,1\}$ is finite rank operator,
the essential spectrum of $U_\omega$ is also $\Sigma$.
\end{remark}

Based on Theorem \ref{thm31} and Remark \ref{remark3}, the eigenvalues and the essential spectrum of $U_\omega$ are illustrated in the complex plane as in Figure1.

\begin{figure}[H]
\begin{center}
\includegraphics[width=0.80\columnwidth]{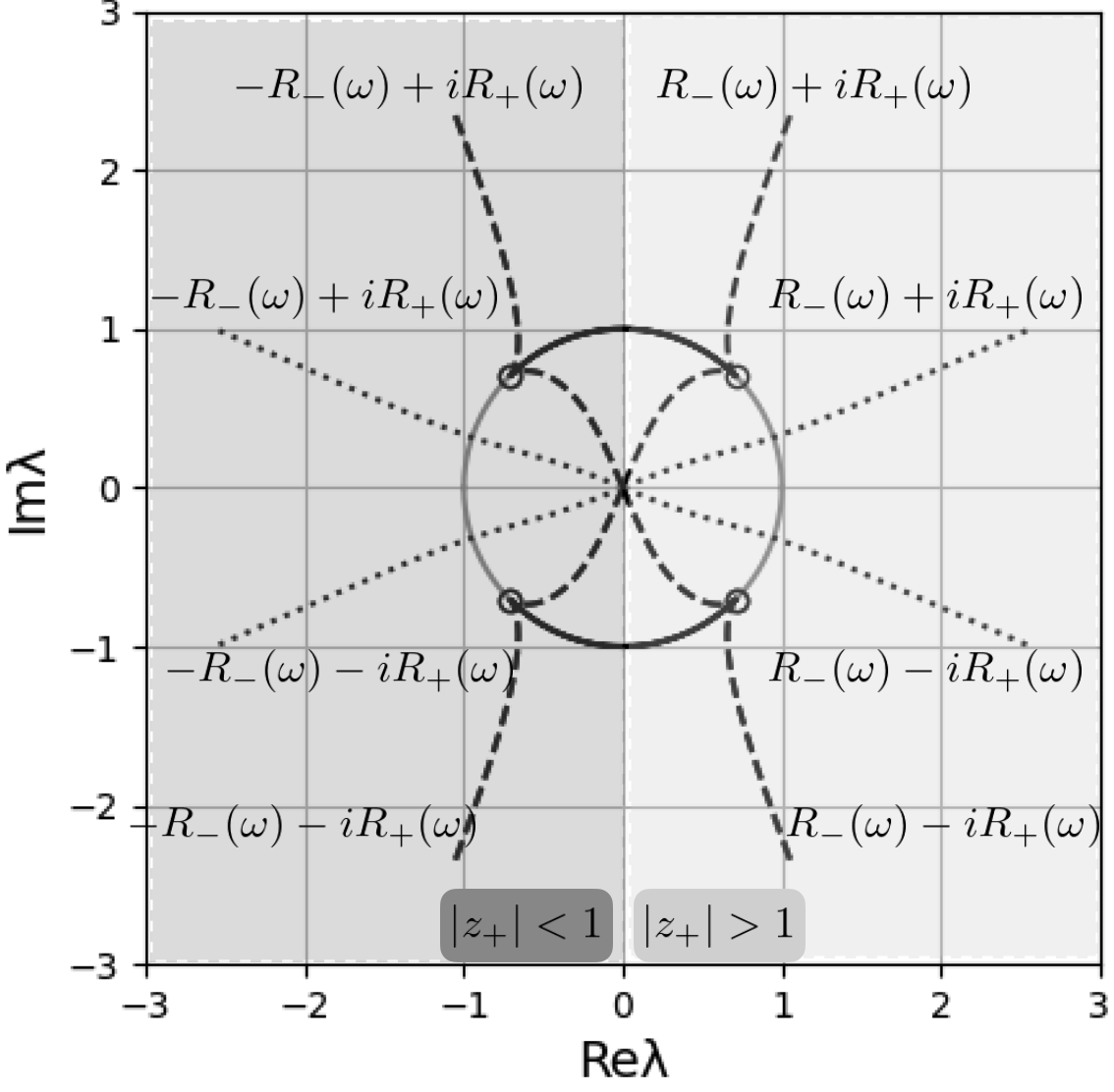}
\caption{{The eigenvalues and essential spectrum of $U_{\omega}$.
($\mathrm{i}$) Dashed line: $\omega > 0$ case;
($\mathrm{ii}$) Dotted line: $\omega < 0$ case;
($\mathrm{iii}$) Black line: the essential spectrum;
($\mathrm{iv}$) Gray line:  the unit circle.
{\label{931769}}%
}}
\end{center}
\end{figure}

\begin{remark}
The points of the intersection of the dotted line and the gray line in Fig. \ref{931769} are eigenvalues of $U_{-1}$, i.e., 
\begin{eqnarray}
\lambda_{1}=\frac{3}{\sqrt{10}}+\frac{i}{\sqrt{10}},\quad 
\lambda_{2}=-\frac{3}{\sqrt{10}}+\frac{i}{\sqrt{10}},\quad 
\lambda_{3}=-\frac{3}{\sqrt{10}}-\frac{i}{\sqrt{10}},\quad 
\lambda_{4}=\frac{3}{\sqrt{10}}-\frac{i}{\sqrt{10}}. \nonumber
\end{eqnarray}
These eigenvalues are in the unit circle, since $U_{-1}$ is unitary.
\end{remark}

\section{Acknowledgements}
I would also like to take this opportunity to thank D. Funakawa, K.Saito, Y.Tanaka, M.Seki, C.Kiumi of collaboration and advice. 

TE is supposed by the JSPS Grant-in-aid for Research Activity Start-up No.23K19004. YM is supported by JST SPRING, Grant Number JPMJSP2144 (Shinshu University).
{\label{687807}}

\section*{Appendix}
In Appendix, we give the proof of Lemma \ref{lemma43}.

($\mathrm{i}$) By definition, $z_+ = z_-$ if and only if $\lambda^2 +\lambda^{-2}=0$. Since $\lambda \neq 0$, by solving the equation $\lambda^4 = -1$, we obtain the assertion.

($\mathrm{ii}$) First, we assume $|z_+|=1$ and derive $\lambda \in \Sigma$. Let $\frac{\lambda +\lambda ^{-1} +\sqrt{\lambda ^2 +\lambda ^{-2}}}{\sqrt{2}}=e^{i\theta} \quad (\theta\in \mathbb{R})$. Rearranging this equation for $\lambda$, we get the following expression:
\begin{align*}
    \lambda^2 -\sqrt{2}\lambda \cos \theta +1=0.
\end{align*}
Here, solving the above quadratic equation for $\lambda$, we obtain 
\begin{align*}
    \lambda = \frac{\cos \theta}{\sqrt{2}}\pm i\frac{\sqrt{1+\sin ^2 \theta}}{\sqrt{2}}
\end{align*}
which is in $\Sigma$.

Conversely, we assume $\lambda \in \Sigma$. Then $\lambda$ is represented as $\lambda = e^{i \theta}$ for some $\theta \in \left[\frac{\pi}{4}, \frac{3\pi}{4} \right] \cup \left[\frac{5\pi}{4}, \frac{7\pi}{4} \right]$. Considering $\cos 2\theta \leq 0$, we can calculate as
\begin{align*}
    z_+=\sqrt{2}\cos \theta +i\sqrt{|\cos 2\theta|}.
\end{align*}
Therefore, summing the squares of the real part and the imaginary part of $z_+$, we find $|z_{+}|=1$.

($\mathrm{iii}$) We use the continuity of the functions to prove this. The square root function $\sqrt{\, \cdot\,}$ is continuous on the region $\mathbb{C} \backslash \mathbb{R}_-$, where $\mathbb{R}_-$ is the set of all negative numbers and zero. Hence, we first determine $\lambda$ such that $\lambda^2 +\lambda^{-2} \in \mathbb{R}_-$. Let $\lambda^2 +\lambda^{-2}=-t \quad (t \geq 0)$. By solving this equation, we have
\begin{align*}
    \lambda ^2 = \frac{-t\pm \sqrt{t^2 -4}}{2}.
\end{align*}
When $0 \leq t\leq 2$, $\lambda^2$ is written as
\begin{align*}
    \lambda^2 =\frac{-t\pm i\sqrt{4-t^2}}{2}.
\end{align*}
Thus, $\lambda^2 \in \mathbb{T}$ and $\Re \lambda ^2 \leq 0$, and we obtain $\lambda \in \Sigma$.
When $t>2$, $\lambda^2$ is written as
\begin{align*}
    \lambda^2 = \frac{-t\pm \sqrt{t^2-4}}{2}<0.
\end{align*}
Therefore, $\lambda \in i \mathbb{R}\backslash \{0 \}$.
Conversely, if $\lambda \in (i\mathbb{R} \backslash \{ 0 \}) \cup \Sigma$, it is easy to prove $\lambda^2 +\lambda^{-2}\in \mathbb{R}_-$.
Consequently, we obtain $\lambda^2 +\lambda^{-2} \in \mathbb{R}_-$ if and only if $\lambda \in (i\mathbb{R} \backslash \{ 0 \}) \cup \Sigma$.
Therefore, if we consider $z_+ = \frac{\lambda+\lambda^{-1} + \sqrt{\lambda^2 + \lambda^{-2}}}{\sqrt{2}}$ as a function of $\lambda$, it is continuous on the region $\mathbb{C}\backslash (i\mathbb{R} \cup \Sigma)$. Moreover, $|z_+|=1$ if and only if $\lambda \in \Sigma$.

Let $\lambda =1$, then $z_+ =1+\sqrt{2}$ is greater than 1. 
By continuity of the function $z_+$, $|z_+|$ is greater than 1 if $\lambda$ is in the open right half-plane excluding $\Sigma$. Similarly, let $\lambda =-1$, then the absolute value of $z_+ = -\sqrt{2} +1$ is less than 1. Hence, $|z_+|$ is less than 1 if $\lambda$ is in the open left half-plane excluding $\Sigma$

Now, the remaining part is $i\mathbb{R}\backslash\{0\}$. 
So, let $\lambda =ik \in i\mathbb{R}\backslash\{0\}$, then
\begin{align*}
    z_{+}= i\frac{k-\frac{1}{k}+\sqrt{k^2+\frac{1}{k^2}}}{\sqrt{2}}
\end{align*}
Thus, $|z_{+}|>1$ when $k>1$ or $-1<k<0$, and $|z_{+}|<1$ when $0<k<1$ or $k<-1$.
This completes the proof of (iii).


\bibliographystyle{unsrtnat}   
\bibliography{mybibfile_R4}

\end{document}